\documentclass[journal]{IEEEtran}
\ifCLASSOPTIONcompsoc
  \usepackage[nocompress]{cite}
\else
  \usepackage{cite}
\fi
\usepackage{amssymb}
\usepackage{amsmath}
\usepackage{dsfont}
\usepackage{amsthm}
\usepackage{graphicx}
\usepackage{mathrsfs}
\usepackage{color}

\newtheoremstyle{mythm}{1.5ex plus 1ex minus .2ex}{1.5ex plus 1ex minus .2ex} {\rm}{\parindent}{\it\it}{\rm{:}}{1em}{}
\theoremstyle{mythm}

\renewcommand{\IEEEQED}{\IEEEQEDopen}

\newtheorem{thm}{Theorem}[section]

\newtheorem{lem}{Lemma}[section]
\newtheorem{prop}{Proposition}[section]
\newtheorem{defn}{Definition}[section]

\ifCLASSINFOpdf
\else
\fi
\begin{document}
%
\title{Noise-induced synchronization of Hegselmann-Krause dynamics in full space}

%
%
%
%

\author{Wei Su, Jin Guo, Xianzhong Chen,  Ge Chen, ~\IEEEmembership{Senior Member,~IEEE}
\IEEEcompsocitemizethanks{\IEEEcompsocthanksitem
This work is supported by the National Key Basic Research Program
of China (973 program) under the grant 2016YFB0800404, the National Natural Science Foundation of China under grants Nos. 61803024, 61671054, 61773054, 11688101, 91427304, the Natural Science Foundation of Beijing under grants Nos. 4182038,  the Fundamental Research Funds for the Central Universities under the grant No. FRF-TP-17-087A1.
\IEEEcompsocthanksitem
W. Su, J. Guo and X. Chen are with the School of Automation and Electrical Engineering, University of Science and Technology Beijing \& Key Laboratory of Knowledge Automation for Industrial Processes, Ministry of Education, Beijing 100083, China, {\tt
suwei@amss.ac.cn, guojin@amss.ac.cn, cxz@ustb.edu.cn}. G. Chen is with the National Center for Mathematics and Interdisciplinary Sciences \& Key Laboratory of Systems and
Control, Academy of Mathematics and Systems Science, Chinese Academy of Sciences, Beijing 100190,
China, {\tt chenge@amss.ac.cn}
}}

\IEEEtitleabstractindextext{%
\begin{abstract}
The Hegselmann-Krause (HK) model is a typical self-organizing system with local rule dynamics. In spite of its widespread use and numerous extensions, the underlying theory of its synchronization induced by noise still needs to be developed. In its original formulation, as a model first proposed to address opinion dynamics, its state-space was assumed to be bounded, and the theoretical analysis of noise-induced synchronization for this particular situation has been well established.
However, when system states are allowed to exist in an unbounded space, mathematical difficulties arise whose theoretical analysis becomes non-trivial and is as such still lacking. In this paper, we completely resolve this problem by exploring the topological properties of HK dynamics and by employing the theory of independent stopping time. The associated result in full state-space provides a solid interpretation of the randomness-induced synchronization of self-organizing systems.
\end{abstract}

}

\maketitle

\IEEEdisplaynontitleabstractindextext

%
\IEEEpeerreviewmaketitle

\ifCLASSOPTIONcompsoc
\IEEEraisesectionheading{\section{Introduction}\label{intro}}
\else
\section{Introduction}\label{intro}
\fi

\IEEEPARstart{I}{n} the past decades, self-organizing systems based on local rules have been used to investigate the collective behavior in natural and social systems, and several models have been proposed, including the widely known Boid and Vicsek models \cite{Reynolds1987,Vicsek1995}. One of the central
issues in the study of collective behavior of self-organizing systems is synchronization. Due to the difficulty of its analysis, most previous theoretical studies on synchronization of self-organizing systems largely ignored the influence of noise \cite{Jad2003, Savkin2004,Tang2007,Chen2014}. However, as Sagu\'{e}s \emph{et al.} explained in \cite{Sagues2007}, ``natural systems are undeniably subject to random fluctuations, arising from either environmental variability or thermal effects''. Moreover, after Heinz von Foerster proposed the principle of ``order from noise'' in 1960 \cite{Foerster1960}, noise has been believed to be a key factor in promoting the synchronization of self-organizing systems, which has been verified in earlier simulation studies \cite{Vicsek1995,HadStauffSchulz2008,HadStauffHan2015}. The corresponding mathematical analyses are, however, only more recent. For example, the analysis of the Vicsek model subject to noise was first carried out by Chen in 2017 \cite{Chen2015}. Beyond that, to the best of our knowledge, a substantial mathematical study on how noise affects the synchronization of self-organizing systems has been infrequent, though admittedly, considerable attention has never ceased to exist in a number of fields \cite{Sagues2007,Shinbrot2001,Matsumoto1983,Eldar2010,Tsimring2014,Zhou2005,Lichtenegger2016,Guo2017,Shirado2017}.

Very recently, we established a theoretical analysis of noise-induced synchronization based on the widely known Hegselmann-Krause (HK) model of opinion dynamics \cite{Su2016}. In the HK model, each agent possesses a bounded confidence and updates its opinion value by averaging the opinions of its neighbors who are located within its confidence region. In spite of its seeming simplicity, the HK model captures a quite fundamental local rule of evolution which is embodied ubiquitously in self-organizing systems, such as the Boid and Vicsek models, and has been largely explored in its deterministic version \cite{Hegselmann2002,Lorenz2005,Krause2000,Etesami2015}. In \cite{Su2016}, via a rigorous analysis, we established for the first time that random noise can enable the HK system to reach synchronization (called \emph{quasi-synchronization} due to noise), and we also obtained a ``critical'' noise strength for quasi-synchronization. Subsequently, Su and Yu \cite{Su2017free} analyzed a truth-seeking HK model with environmental noise, and proved that even a small amount of noise can drive all agents in the system towards a state of attained truth.

The analysis of noisy HK models in previous studies was subject to an assumption that all agents' opinions were limited to a bounded interval \cite{Su2016,Pineda2013,Su2017free,Wang2017,Garnier2017}. In particular, the boundedness assumption was crucial to the proof of noise-induced synchronization of HK dynamics in \cite{Su2016}. When the state-space is bounded, the system has a \emph{uniform} positive probability to reach quasi-synchronization in a finite period from \emph{any} initial state. However, when the state of a noisy HK model is allowed to exist in the full space, the system has no uniform positive probability to attain quasi-synchronization in a finite period from any initial state, leading thus to invalidation of the existing methods.

In this paper, via exploring the topological property of a noisy HK system in full state-space, and using the theory of independent stopping time, we show that for any initial state, the system will reach a state whose neighboring graph consists of all complete subgraphs, with a uniform positive probability in a finite period. Additionally, given any initial state whose neighboring graph consists of all complete subgraphs, we prove that the system will achieve quasi-synchronization with a uniform positive probability in an almost surely (a.s.) finite stopping time. Combining the two conclusions leads us to the final answer. Importantly, we wish to stress here that finding the uniform positive probability in a finite stopping time is essentially a new skill which may extend the idea of ``joint connectivity in a finite period'' to the ``joint connectivity in a finite stopping time'' in the consensus of multi-agent systems.

Besides this novel mathematical achievement, another highlighting contribution of this paper is its physical significance in providing a theoretical interpretation of the noise-induced synchronization of self-organizing systems. Though the HK model with bounded state-space performs generally well in mimicking opinion behavior, the imposed assumption of boundedness of the state-space is undesirable in physical systems and is largely limiting its potential for representing an elementary self-organizing system.

The rest of the paper is organized as follows: Section \ref{formulation} presents some preliminaries of the underlying model; in Section \ref{Results}, we give the main results of the paper; Section \ref{Simulations} shows simulation results that verify the main theoretical conclusions, and finally, some concluding remarks are given in Section \ref{Conclusions}.

\renewcommand{\thesection}{\Roman{section}}
\section{Model and definitions}\label{formulation}
\renewcommand{\thesection}{\arabic{section}}
Denote $\mathcal{V}=\{1,2,\ldots,n\}$ as the set of $n$ agents, $x_i(t)\in(-\infty,\infty), i\in\mathcal{V}, t\geq 0$ as the state of agent $i$ at time $t$. The update rule of HK dynamics then takes:
\begin{equation}\label{basicHKmodel}
  x_i(t+1)=\frac{1}{|\mathcal{N}_i(x(t))|}\sum\limits_{j\in\mathcal{N}_i(x(t))}x_j(t)+\xi_i(t+1),\,\,\,i\in\mathcal{V},
\end{equation}
where
\begin{equation}\label{neigh}
 \mathcal{N}_i( x(t))=\{j\in\mathcal{V}\; \big|\; |x_j(t)-x_i(t)|\leq \epsilon\}
\end{equation}
is the neighbor set of $i$ at $t$, $\epsilon>0$ represents the confidence threshold of the agents and $\xi_i(t), i\in\mathcal{V}, t\geq 1$ is noise. Here, $|\cdot|$ can be the cardinal number of a set or the absolute value of a real number.

In \cite{Su2016}, the state-space is assumed to be bounded, i.e. $x_i(t)\in[0,1], i\in\mathcal{V}, t\geq 0$. If there is no noise, it is proved that for any given initial opinion value $x(0)\in[0,1]^n$, the evolutionary opinion values $x(t), t\geq 0$ of the noise-free HK model cannot exceed the initial boundary opinions. However, in the presence of noise, mathematically, the evolutionary opinion values can be driven to run outside the initial boundary opinions, and even outside the opinion space $[0,1]$. In \cite{Su2016}, to limit the noisy opinion values in $[0,1]$, it forcibly assumes that $x_i(t+1)=0$ or 1 when $\frac{1}{|\mathcal{N}_i(x(t))|}\sum\limits_{j\in\mathcal{N}_i(x(t))}x_j(t)+\xi_i(t+1)$ is less than 0 or larger than 1. In model (\ref{basicHKmodel}), this assumption is cancelled and the state-space is allowed to be unbounded.

To proceed, some preliminary definitions are first needed.
\begin{defn}\label{graphdef}
Let $\mathcal{G}_\mathcal{V}(t)=\{\mathcal{V},\mathcal{E}(t)\}$ be the graph of $\mathcal{V}$ at time $t$, and $(i,j)\in\mathcal{E}(t)$ if and only if $|x_i(t)-x_j(t)|\leq \epsilon$. A graph $\mathcal{G}_\mathcal{V}(t)$ is called a {\it complete graph} if and only if $(i,j)\in\mathcal{E}(t)$ for any $i,j\in\mathcal{V}$; and $\mathcal{G}_\mathcal{V}(t)$ is called a {\it connected graph} if and only if for any $i\neq j$, there are edges $(i,i_1), (i_1,i_2), \ldots, (i_k,j)$ in $\mathcal{E}(t)$.
\end{defn}
The definition of \emph{quasi-synchronization} of the noisy model (\ref{basicHKmodel})-(\ref{neigh}) is given by \cite{Su2016}:
\begin{defn}\label{robconsen}
Denote
\begin{equation*}\label{opindist}
  d_{\mathcal{V}}(t)=\max\limits_{i, j\in \mathcal{V}}|x_i(t)-x_j(t)|~~\mbox{and}~~d_{\mathcal{V}}=\limsup\limits_{t\rightarrow \infty}d_{\mathcal{V}}(t).
\end{equation*}
(i) if $d_{\mathcal{V}} \leq \epsilon$, we say the system (\ref{basicHKmodel})-(\ref{neigh}) will reach quasi-synchronization.\\
(ii) if $P\{d_{\mathcal{V}} \leq \epsilon\}=1$, we say almost surely (a.s.) the system (\ref{basicHKmodel})-(\ref{neigh}) will reach quasi-synchronization.\\
(iii) if  $P\{d_{\mathcal{V}} \leq \epsilon\}=0$, we say a.s. the system (\ref{basicHKmodel})-(\ref{neigh}) cannot reach quasi-synchronization.\\
(iv) let $T=\min\{t: d_{\mathcal{V}}(t')\leq \epsilon \mbox{ for all } t'\geq t\}$.
 If $P\{T<\infty\}=1$, we say a.s. the system (\ref{basicHKmodel})-(\ref{neigh}) reaches quasi-synchronization in finite time.
\end{defn}

\renewcommand{\thesection}{\Roman{section}}
\section{Main Results}\label{Results}
\renewcommand{\thesection}{\arabic{section}}
For simplicity, we first present the result for quasi-synchronization with independent and identically distributed (i.i.d.) noises, which can be directly derived from the two subsequent general results with independent noises.
\begin{thm}[Critical noise amplitude for quasi-synchronization of HK model with i.i.d. noise]\label{consthm0}
Let $\{\xi_i(t)\}_{i\in\mathcal{V},t\geq 1}$ be non-degenerate random variables with independent and identical distribution, then for any $x(0)\in (-\infty,\infty)^n$ and $\epsilon>0$,
 \begin{enumerate}
   \item[(i)] if there exist constants $\delta_1,\delta_2\in(-\infty,\infty)$ with $\delta_2-\delta_1=\epsilon$ such that $P\{\delta_1\leq\xi_1(1)\leq \delta_2\}=1$, then a.s. the system (\ref{basicHKmodel})-(\ref{neigh}) will reach quasi-synchronization in finite time;
   \item[(ii)]  if
    $P\{\delta_1\leq\xi_1(1)\leq \delta_2\}<1$ for any $\delta_2-\delta_1=\epsilon$,
then a.s. the system (\ref{basicHKmodel})-(\ref{neigh}) cannot reach quasi-synchronization.
 \end{enumerate}
\end{thm}
Conclusion (i) shows that if noise amplitude is no more than $\epsilon$, the system will a.s. achieve quasi-synchronization in finite time; Conclusion (ii) states that when noise amplitude has a positive probability to exceed $\epsilon$, the system will not reach quasi-synchronization. This implies $\epsilon$ is the critical noise amplitude to induce a quasi-synchronization. Conclusions (i) and (ii) can be directly derived from the following Propositions \ref{Suff_thm} and \ref{Nece_thm}, which present sufficient and necessary conditions, respectively, for independent noises:

\begin{prop}[Sufficient condition for quasi-synchronization of HK model with independent noise]\label{Suff_thm}
Let $\{\xi_i(t), i\in\mathcal{V},t\geq 1\}$ be independent random variables with $E\xi_i(t)=C\in(-\infty,\infty)$ and satisfy:
 i) $P\{\delta_1\leq\xi_i(t)\leq \delta_2\}=1$ with $0<\delta_2-\delta_1\leq \epsilon$;
 ii) there exist constants $a\in(0,\frac{\delta_2-\delta_1}{2}),p\in(0,1)$ such that
$P\{\xi_i(t)\geq \frac{\delta_2+\delta_1}{2}+a\}\geq p$ and $P\{\xi_i(t)\leq \frac{\delta_2+\delta_1}{2}-a\}\geq p$.
Then, for any initial state $x(0)\in (-\infty,\infty)^n$ and  $\epsilon>0$, the system (\ref{basicHKmodel})-(\ref{neigh}) will a.s. reach quasi-synchronization in finite time and $d_\mathcal{V}\leq \delta_2-\delta_1$ a.s.
\end{prop}

\noindent{\it Proof of Theorem \ref{consthm0} (i):} Noting that for i.i.d. random variables $\{x_i(t), i\in\mathcal{V}, t\geq 1\}$ with $Ex_1(1)=0,\,Ex_1^2(1)>0$, there exist constants $a>0$ and $0<p\leq 1$, such that
\begin{equation*}
  P\{x_i(t)> a\}\geq p,\,\,\,P\{x_i(t)< -a\}\geq p,
\end{equation*}
the conditions in Proposition \ref{Suff_thm} can be satisfied. \hfill $\Box$

To prove Proposition \ref{Suff_thm}, some lemmas are need:
\begin{lem}\cite{Krause2000}\label{monosmlem}
Suppose $\{z_i, \, i=1, 2, \ldots\}$ is a nonnegative nondecreasing (nonincreasing) sequence.  Then for any $s\geq 0$, the sequence
$\{g_s(k)=\frac{1}{k}\sum_{i=s+1}^{s+k}z_i$, $k\geq 1\}$ is monotonically nondecreasing (nonincreasing) for $k$.
\end{lem}
In what follows, let $\xi(t)=\{\xi_i(t), i\in\mathcal{V}\}$ and the ever appearing time symbols $t$ (or $T$, etc.) all refer to the random variables $t(\omega)$ (or $T(\omega)$, etc.) on the probability space $(\Omega,\mathcal{F},P)$, and will be still written as $t$ (or $T$, etc.) for simplicity.
\begin{lem}\cite{Su2016}\label{robconspeci}
For the system (\ref{basicHKmodel})-(\ref{neigh}) with conditions of Proposition \ref{Suff_thm} i), if there exists a finite time $0\leq T<\infty$ such that $d_\mathcal{V}(T)\leq \epsilon$, then we have
$d_\mathcal{V}(t)\leq \delta_2-\delta_1$ for $t>T$.
\end{lem}
The following lemma is key to obtain the \emph{uniformly} positive probability of reaching quasi-synchronization from \emph{any} initial states:
\begin{lem}\label{clusterconsen}
For system (\ref{basicHKmodel})-(\ref{neigh}) with conditions of Proposition \ref{Suff_thm} i), if there exists a finite time $ T$ and disjoint subsets $\mathcal{V}_k\subset\mathcal{V}, k=1,\ldots,m (1\leq m\leq n)$ such that $d_{\mathcal{V}_k}(T)\leq \epsilon, 1\leq k\leq m$, and for $k_1\neq k_2, \mathcal{V}_{k_1}\bigcap\mathcal{V}_{k_2}=\emptyset$, $|x_i(T)-x_j(T)|>\epsilon, i\in\mathcal{V}_{k_1}, j\in\mathcal{V}_{k_2}$, then there exist constants $0<p_0\leq 1, L_0>0$ and a finite stopping time series $T_i$ which is $\sigma(\xi((i-1)L_0+T+\sum_{j=1}^{i-1}T_j)+1,\ldots)-measurable, i=1,\ldots,m-1$ such that
$P\{d_\mathcal{V}(T+(m-1)L_0+T_1+\ldots+T_{m-1})\leq \delta_2-\delta_1\}\geq p_0$.
\end{lem}
\begin{proof}
Without loss of generality, suppose $T=0$ a.s. Then at the initial moment, the system forms $m$ subgroups with complete graphs, and by (\ref{basicHKmodel}), $d_{\mathcal{V}_k}(1)\leq \delta_2-\delta_1\leq \epsilon, k=1,\ldots,m$. Before one subgroup enters the neighbor region of another, for each $i\in\mathcal{V}_k, \,1\leq k\leq m$, we have
\begin{equation}\label{xitvk}
\begin{split}
  &x_i(t+1)=\frac{1}{|\mathcal{V}_k|}\sum\limits_{j\in\mathcal{V}_k}x_j(t)+\xi_i(t+1)\\
  =&\frac{1}{|\mathcal{V}_k|}\sum\limits_{j\in\mathcal{V}_k}x_j(0)+\sum\limits_{l=1}^{t}\frac{\sum\limits_{j\in\mathcal{V}_k}\xi_j(l)}{|\mathcal{V}_k|}+\xi_i(t+1).
  \end{split}
\end{equation}
Order the subgroups at any moment $t\geq 1$ by the state values as $1,2,\ldots,m$, and consider the subgroups $\mathcal{V}_1(1)$ with smallest state values and $\mathcal{V}_m(1)$ with the largest state values.
For $t\geq 0, \,k=1,\ldots,m$, let $$y_k(t+1)=\frac{1}{|\mathcal{V}_k|}\sum\limits_{j\in\mathcal{V}_k}x_j(0)+\sum\limits_{l=1}^{t}\frac{\sum\limits_{j\in\mathcal{V}_k}\xi_j(l)}{|\mathcal{V}_k|}.$$
Then for $t\geq 1$,
\begin{equation*}
\begin{split}
  y_m(t)-y_1(t)=&\frac{\sum_{j\in\mathcal{V}_m}x_j(0)}{|\mathcal{V}_m|}-\frac{\sum_{j\in\mathcal{V}_1}x_j(0)}{|\mathcal{V}_1|}\\
  &+\sum\limits_{l=1}^{t-1}\bigg(\frac{\sum_{j\in\mathcal{V}_m}\xi_j(l)}{|\mathcal{V}_m|}-\frac{\sum_{j\in\mathcal{V}_1}\xi_j(l)}{|\mathcal{V}_1|}\bigg).
  \end{split}
\end{equation*}
Since $\xi(t)=\{\xi_i(t), i\in\mathcal{V}\}, t\geq 1$ are independent, the $\sigma$-algebras $\sigma(\xi(t)), t\geq 1$ are independent. By Law of the Iterated Logarithm (Theorem 10.2.1 of \cite{Chow1997}), we have that
\begin{equation}\label{ytlimits}
\begin{split}
  \limsup\limits_{t\rightarrow\infty}(y_m(t)-y_1(t))&=\infty,\,\,a.s., \\
   \liminf\limits_{t\rightarrow\infty}(y_m(t)-y_1(t))&=-\infty,\,\,a.s.
   \end{split}
\end{equation}
Notice that $\Big|\frac{\sum_{j\in\mathcal{V}_m}\xi_j(l)}{|\mathcal{V}_m|}-\frac{\sum_{j\in\mathcal{V}_1}\xi_j(l)}{|\mathcal{V}_1|}\Big|\leq \delta_2-\delta_1\leq \epsilon$ a.s., by (\ref{ytlimits}), there exists a $\sigma_t$-time $0\leq T_0<\infty$ where $\sigma_t=\sigma(\xi(1),\ldots,\xi(t))$ that
\begin{equation}\label{ytepsi}
  0<y_m(T_0)-y_1(T_0)\leq \epsilon,\,\,a.s.
\end{equation}
Combining (\ref{xitvk}) and (\ref{ytepsi}), we obtain that there a.s. exists a $\sigma_t$-time $T_1\leq T_0$ such that at $T_1$, at least two subgroups with complete graphs will for the first time enter the neighbor region of one another and become a new complete or connected graph. Denote $\widetilde{x}_i(t)=|\mathcal{N}(i, x(t))|^{-1}\sum_{j\in \mathcal{N}(i, x(t))}x_j(t)$, $i\in\mathcal{V}, t\geq 0$ and let $\mathcal{V}_c(t)$ be the new emerging subgroups with connected but not complete graphs at $t$, then for $i\in\mathcal{V}_c(t)$ design the following protocol:
\begin{equation}\label{noiseproto0}
  \left\{
    \begin{array}{ll}
      &\xi_i(t+1)\in[\frac{\delta_2+\delta_1}{2}+a,\delta_2], \quad\hbox{if}\\
      &\,\, \min\limits_{j\in\mathcal{V}_c(t)}x_j(t)\leq\widetilde{x}_i(t)\leq \min\limits_{j\in\mathcal{V}_c(t)}x_j(t)+\frac{d_{\mathcal{V}_c(t)}(t)}{2}; \\
      &\xi_i(t+1)\in[\delta_1,\frac{\delta_2+\delta_1}{2}-a], \quad\hbox{if}\\
      &\,\,\min\limits_{j\in\mathcal{V}_c(t)}x_j(t)+\frac{d_{\mathcal{V}_c(t)}(t)}{2}<\widetilde{x}_i(t)\leq \max\limits_{i\in\mathcal{V}_c(t)}x_j(t).
    \end{array}
  \right.
\end{equation}
For all $\mathcal{V}_c(t)$, by (\ref{basicHKmodel}) and Lemma \ref{monosmlem}, we know that under protocol (\ref{noiseproto0}) the minimum state value of $\mathcal{V}_c(t)$ increases by at least $a$, the maximum state value of $\mathcal{V}_c(t)$ decreases by at least $a$, and $d_{\mathcal{V}_c(t)}(t)$ reduces by at least $2a$ after each step.
Moreover, we know that $d_{\mathcal{V}_c(t)}(t)\leq n\epsilon$, then under the protocol (\ref{noiseproto0}), there must exist a constant $L_0\leq \lceil\frac{(n-1)\epsilon}{2a}\rceil$ such that $d_{\mathcal{V}_c(T_1+L_0)}(T_1+L_0)\leq \epsilon$ a.s. (This also means protocol (\ref{noiseproto0}) occurs $L_0$ times).  Since there exist $m\leq n$ subgroups with complete graphs at the initial moment, by following the above procedure, we obtain that under the protocol (\ref{noiseproto0}), the whole group $\mathcal{V}$ will a.s. form a complete graph in a finite time $\bar{T}\leq \sum_1^{m-1}T_j+(m-1)L_0$ where $T_j$ is $\sigma(\xi((j-1)L_0+\sum_1^{i-1}T_j+1),\ldots)$-measurable, and during this process, protocol (\ref{noiseproto0}) occurs no more than $(m-1)L_0$ times. By independence of $\xi_i(t), i\in\mathcal{V}, t\geq 1$, we know that
$$
  P\{\text{protocol (\ref{noiseproto0}) occurs}~ (m-1)L_0 ~\text{times}\}\\
  \geq p^{n(m-1)L_0}>0.
$$
Let $p_0=p^{n(n-1)L_0}$ and consider Lemma \ref{robconspeci}, then we obtain the conclusion.
\end{proof}

\noindent{\it Proof of Proposition \ref{Suff_thm}:} For each $t\geq 0$ and any given $x(t)\in(-\infty,\infty)^n$, it is easy to check that there exist disjointed subsets $\mathcal{V}_k(t), k=1,\ldots,m (1\leq m\leq n)$ such that $\mathcal{V}=\bigcup_1^m\mathcal{V}_k(t)$ and each $\mathcal{G}_{\mathcal{V}_k}(t)=\{\mathcal{V}_k(t),\mathcal{E}_k(t)\}$ is either a complete graph or a connected but not complete graph. If $G_{\mathcal{V}}(0)$ is a complete graph, by Lemma \ref{robconspeci}, the conclusion holds. Otherwise, at each moment $t\geq 0$ consider the protocol (\ref{noiseproto0}) for all subsets $\mathcal{V}_k(t), 1\leq k\leq m$ with connected but not complete graphs.

If $\mathcal{G}_k(t)$ is a connected but complete graph, following the same argument below (\ref{noiseproto0}), we know that $d_{\mathcal{V}_c(t)}(t)$ reduces by at least $2a$ after each step under the protocol (\ref{noiseproto0}). If $\mathcal{G}_k(t)$ is a complete graph, by Lemma \ref{robconspeci}, we know that $d_{\mathcal{V}_c(t)}(t)\leq \epsilon$ before it meets another subgroup.
Since $d_{\mathcal{V}_k(t)}(t)\leq |\mathcal{V}_k(t)|\epsilon\leq n\epsilon$ when $\mathcal{G}_k(t)$ is a connected but not complete graph, we can get that under protocol (\ref{noiseproto0}), $\mathcal{G}_k(t)$ will become a complete graph after no more than $\lceil\frac{(n-1)\epsilon}{2a}\rceil$ steps. Considering that during this period two subgroups may meet and become a new connected but not a complete graph, we know that under the protocols (\ref{noiseproto0}), all subgroups will become complete graphs after no more than $\lceil\frac{(n-1)^2\epsilon}{2a}\rceil$ steps. By independence of $\xi_i(t), i\in\mathcal{V}, t>0$, we know that
\begin{equation*}
\begin{split}
 &P\Big\{\text{protocol}~ (\ref{noiseproto0})~ \text{occurs}~ \Big\lceil\frac{(n-1)^2\epsilon}{2a}\Big\rceil~\text{times}\Big\}\\
 \geq& p^{\lceil\frac{n(n-1)^2\epsilon}{2a}\rceil}>0,
\end{split}
\end{equation*}
implying for any given $x(0)\in(-\infty,\infty)^n$, there exists a constant $L\leq \lceil\frac{(n-1)^2\epsilon}{2a}\rceil$ such that
\begin{equation}\label{clusterfinitime}
\begin{split}
  &P\{G_\mathcal{V}(L)~\text{consists of complete graphs}\}\\
  \geq& p^{\lceil\frac{n(n-1)^2\epsilon}{2a}\rceil}>0.
  \end{split}
\end{equation}
Denote $C(L)=\{\omega:G_\mathcal{V}(L)~\text{consists of complete graphs}\}$, then by Lemma \ref{clusterconsen}, there exists a finite time $\bar{T}_1$ which is $\sigma(\xi(1),\ldots)$-measurable, and a constant $0<p_0<1$ such that
\begin{equation*}
\begin{split}
  P\{d_\mathcal{V}(L+\bar{T}_1)\leq\epsilon\}&=P\{d_\mathcal{V}(\bar{T}_1)\leq \epsilon|C(L)\}\cdot P\{C(L)\}\\
  &\geq p_0p^{\lceil\frac{n(n-1)^2\epsilon}{2a}\rceil}>0,
  \end{split}
\end{equation*}
and hence
\begin{equation}\label{probnocons}
  P\{d_\mathcal{V}(L+\bar{T}_1)>\epsilon\}\leq 1-p_0p^{\lceil\frac{n(n-1)^2\epsilon}{2a}\rceil}<1.
\end{equation}
For a finite time $T$, define $U(T)=\{\omega: d_\mathcal{V}(L+T)>\epsilon\}, U=\{\omega: (\ref{basicHKmodel})-(\ref{neigh})$ does not reach quasi-synchronization in finite time $\}$.
By (\ref{probnocons}),
$$
  P\{U(\bar{T}_1)\}\leq 1-p_0p^{\lceil\frac{n(n-1)^2\epsilon}{2a}\rceil}<1.
$$
Since $x(0)$ is arbitrarily given in $(-\infty,\infty)^n$, considering the independence of $\sigma(\bar{T}_1)$ and $\sigma(\xi(\bar{T}_1+1),\ldots)$ and following the procedure of (\ref{probnocons}), we know there exists a finite time sequence $\bar{T}_1\leq \bar{T}_2\leq \ldots< \infty$ such that
\begin{equation*}
  P\{U(\bar{T}_{m+1})|U(\bar{T}_m)\}\leq 1-p_0p^{\lceil\frac{n(n-1)^2\epsilon}{2a}\rceil},\quad m\geq 1.
\end{equation*}
Notice by Lemma \ref{robconspeci} that once there is a finite time $T$ that $d_\mathcal{V}(T)\leq \epsilon$, it will hold $d_\mathcal{V}\leq \delta_2-\delta_1\leq \epsilon$, thus $U(\bar{T}_{j+1})\subset U(\bar{T}_j), j\geq 1$ and hence
\begin{equation*}
\begin{split}
  P\{U\} \leq &P\Big\{\bigcap\limits_{m=1}^{\infty}U(\bar{T}_m)\Big\}=\lim\limits_{m\rightarrow \infty}P\Big\{\bigcap\limits_{j=1}^{m}U(\bar{T}_j)\Big\}\\
=&\lim\limits_{m\rightarrow \infty}\prod\limits_{j=1}^{m-1}P\Big\{U(\bar{T}_{j+1}\Big|\bigcap\limits_{l\leq j} U(\bar{T}_l)\Big\}\cdot P\{U(\bar{T}_1)\}\\
=&\lim\limits_{m\rightarrow \infty}\prod\limits_{j=1}^{m-1}P\{U(\bar{T}_{j+1}| U(\bar{T}_j)\}\cdot P\{U(\bar{T}_1)\}\\
\leq& \lim\limits_{m\rightarrow \infty}(1-p_0p^{\lceil\frac{n(n-1)^2\epsilon}{2a}\rceil})^m=0,
\end{split}
\end{equation*}
here the first equation holds since $\{\bigcap\limits_{j=1}^{m}U(\bar{T}_j), m\geq 1\}$ is a decreasing sequence and $P$ is a probability measure. As a result
\begin{equation*}
\begin{split}
  & P\{(\ref{basicHKmodel})-(\ref{neigh})~ \text{ reach quasi-synchronization in finite time}\}\\
  =& 1-P\{U\}=1.
\end{split}
\end{equation*}
This completes the proof. \hfill$\Box$

Next, we will present the necessary part of the noise-induced synchronization, which shows that when the noise amplitude has a positive probability of exceeding $\epsilon$, the system a.s. cannot reach quasi-synchronization.
\begin{prop}[Necessary condition for quasi-synchronization of the HK model with independent noise] \label{Nece_thm}
Let $x(0)\in (-\infty,\infty)^n$, $\epsilon>0$ are arbitrarily given.
Suppose the noises $\{\xi_i(t), i\in\mathcal{V},t\geq 1\}$ are independent and there exist constants $\delta_2-\delta_1\geq \epsilon$ and $0<q\leq 1$ such that
    $P\{\xi_i(t)<\delta_1\}\geq q$ and $P\{\xi_i(t)>\delta_2\}\geq q$,
then a.s. the system (\ref{basicHKmodel})-(\ref{neigh}) cannot reach quasi-synchronization.
\end{prop}
\begin{proof}
We only need to prove that, for any $T_0\geq 0$, there exists $t> T_0$ a.s. such that $d_\mathcal{V}(t)>\epsilon$ a.s., i.e.
\begin{equation*}
  P\Big\{\bigcup\limits_{T_0=0}^{\infty}\{d_\mathcal{V}(t)\leq\epsilon, t> T_0\}\Big\}=0.
\end{equation*}
Given any $T_0\geq 0$, by independence of $\xi_i(t), i\in\mathcal{V}, t\geq 1$, it has
\begin{equation*}
\begin{split}
  &P\{d_\mathcal{V}(T_0+1)>\epsilon|d_\mathcal{V}(T_0)\leq\epsilon\}\\
  \geq& P\Big\{\max\limits_{x_i(T_0),i\in\mathcal{V}}\xi_i(T_0+1)>\delta_2, \min\limits_{x_i(T_0),i\in\mathcal{V}}\xi_i(T_0+1)<\delta_1\Big\}\\
  \geq& q^2.
  \end{split}
\end{equation*}
Hence, $P\{d_\mathcal{V}(T_0+1)\leq \epsilon|d_\mathcal{V}(T_0)\leq\epsilon\}\leq 1-q^2<1$. Similarly, for all $t>T_0$
\begin{equation*}
  P\Big\{d_\mathcal{V}(t)\leq\epsilon\Big|\bigcap\limits_{T_0\leq l<t}\{d_\mathcal{V}(l)\leq\epsilon\}\Big\}\leq 1-q^2.
\end{equation*}
Noting $P\{d_\mathcal{V}(T_0)\leq \epsilon\}\leq 1$, it has
\begin{equation*}
  \begin{split}
 & P\{d_\mathcal{V}(t)\leq \epsilon, t> T_0\}=P\Big\{\bigcap\limits_{t=T_0+1}^\infty\{d_\mathcal{V}(t)\leq\epsilon\}\Big\}\\
 =&\lim\limits_{m\rightarrow\infty}P\Big\{\bigcap\limits_{t=T_0+1}^m\{d_\mathcal{V}(t)\leq\epsilon\}\Big\}\\
  \leq&\lim\limits_{m\rightarrow\infty}\prod\limits_{t=T_0+1}^mP\Big\{d_\mathcal{V}(t)\leq \epsilon\Big|\bigcap\limits_{T_0\leq l< t}\{d_\mathcal{V}(l)\leq\epsilon\}\Big\}\\
  \leq &\lim\limits_{m\rightarrow\infty}(1-q^2)^m=0.
  \end{split}
\end{equation*}
This completes the proof.
\end{proof}
Proposition \ref{Nece_thm} shows that when the noise amplitude has a positive probability to exceed the confidence threshold, the cluster will be destroyed by noise. Thus, the essence of noise-induced synchronization is that noise can drive the system towards a synchronized state, and the system is capable of maintaining that state. When the noise is large and exceeds a critical amplitude, the system is fluctuating severely such that the synchronized state cannot be maintained anymore.

\renewcommand{\thesection}{\Roman{section}}
\section{Simulations}\label{Simulations}
\renewcommand{\thesection}{\arabic{section}}
In this part, we will present some simulation results to verify the main theoretical results in this paper. First, we present a fragmentation of noise-free HK model. We take $n= 20, \epsilon=5$ and the initial states are randomly generated on $[0,50]$. Fig. \ref{noise0fig} shows the formation of four clusters. We then add independent noises which are uniformly distributed on $[\delta_1, \delta_2]$ to the agents. By Proposition \ref{Suff_thm}, when $0<\delta_2-\delta_1\leq \epsilon$, the system will almost surely achieve quasi-synchronization. Let $\delta_1=-2, \delta_2=2.1$, then Fig. \ref{noise1fig} clearly displays the quasi-synchronization picture. Next, we consider the case when the noise amplitude $\delta_2-\delta_1$ exceeds the critical value $\epsilon$. For a better demonstration, we simply show a synchronized system will divide in the presence of larger noise. After taking $n=10, x_i(0)=0, 1\leq i\leq 10$ and $\epsilon=5$. Let $\delta_1=-3, \delta_2=3.5$, our Fig. \ref{noise6fig} shows the clear separation of the system.

\begin{figure}[htp]
  \centering
  \includegraphics[width=2.5in]{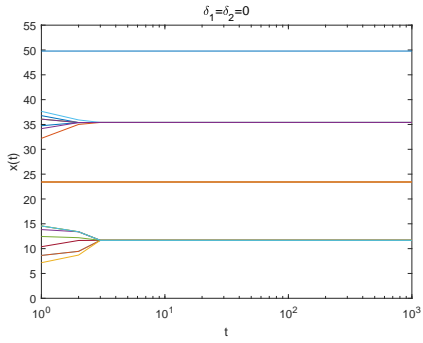}\\
  \caption{Evolution of system (\ref{basicHKmodel})-(\ref{neigh}) of 20 agents without noise. The initial system states are randomly generated on  $[0,50]$, confidence threshold $\epsilon=5$. }\label{noise0fig}
\end{figure}

\begin{figure}[htp]
  \centering
  \includegraphics[width=2.5in]{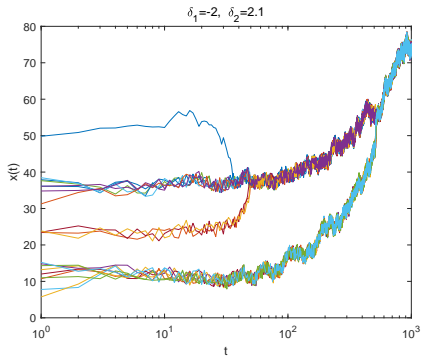}\\
  \caption{Evolution of system (\ref{basicHKmodel})-(\ref{neigh}) of 20 agents with noise uniformly distributed on $[-2, 2.1]$. The initial conditions are identical with those in Fig. \ref{noise0fig}, except that adding noises are uniformly distributed on $[-2, 2.1]$.  }\label{noise1fig}
\end{figure}

\begin{figure}[htp]
  \centering
  \includegraphics[width=2.5in]{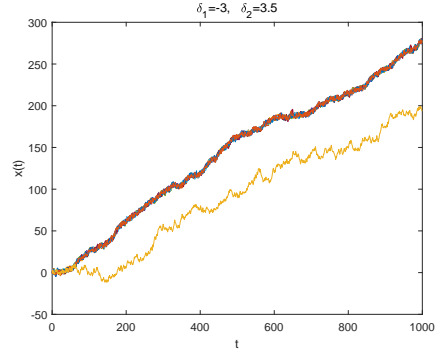}\\
  \caption{Evolution of system (\ref{basicHKmodel})-(\ref{neigh}) of 10 agents with noise uniformly distributed on $[-3,3.5]$. The initial system states are identically taken to be 0; the confidence threshold $\epsilon=5$.  }\label{noise6fig}
\end{figure}

\renewcommand{\thesection}{\Roman{section}}
\section{Conclusions}\label{Conclusions}
\renewcommand{\thesection}{\arabic{section}}
In this paper, we mainly established a rigorous theoretical analysis for noise-induced synchronization of the HK model in the full state-space. By investigating the graph property of the HK dynamics, we completely resolved this problem. Moreover, a critical noise amplitude for the induced synchronization is obtained. The analysis skill that we developed for the graph property of the HK model will provide further tools for studying synchronization problems in noisy HK-based dynamics. Moreover, given the flexible generalizability of our results, we hope our analysis will stimulate much further research on noise-induced synchronization phenomena in physical, biological, and social self-organizing systems.

\end{document}